\documentclass[a4paper,UKenglish,cleveref, autoref, thm-restate]{article}

\usepackage[T1]{fontenc}
\usepackage{latexsym}
\usepackage{color}
\usepackage{mathrsfs}
\usepackage{amsmath,amssymb,latexsym, amsthm}

\usepackage{tikz}

\newcommand{\N}{\mathbb{N}}
\newcommand{\Z}{\mathbb{Z}}

\newcommand{\genF}{\tilde{\omega}(\mathcal{F})}
\newcommand{\genFM}{\tilde{\omega}(\mathcal{F}_M)}

\theoremstyle{plain} \newtheorem{theorem}{Theorem}[section]
\newtheorem{example}[theorem]{Example}
\newtheorem{proposition}[theorem]{Proposition}
\newtheorem{lemma}[theorem]{Lemma}
\newtheorem{claim}[theorem]{Claim}
\newcommand{\bprf}[1][Proof:]{\begin{list}{}    {\setlength{\leftmargin}{0.5em} \setlength{\rightmargin}{0em}  \setlength{\listparindent}{1em}}   \item {\em \hspace{-1em}  #1  }}
\newcommand{\eprf}{\end{list}}
\newcommand{\bclaimproof}{\bprf}

\theoremstyle{definition} 
\newtheorem{definition}{Definition}[section]

\newcommand{\A}{\mathcal{A}}

\newcommand{\Si}{\Sigma}
\newcommand{\SZ}{\Sigma^{\Z}}

\newcommand{\sti}{\vbox to 7pt{\hbox{\includegraphics{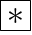}}}\hspace{.1em}}
\newcommand{\sts}{\vbox to 7pt{\hbox{\includegraphics{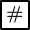}}}\hspace{.1em}}
\newcommand{\stsp}{\vbox to 7pt{\hbox{\includegraphics{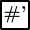}}}\hspace{.1em}}

\title{Rice's theorem for generic limit sets of cellular automata}
\author{Martin Delacourt\\Université d'Orléans, LIFO EA4022, FR-45067 Orléans, France\\martin.delacourt@univ-orleans.fr}

\begin{document}

\maketitle
\begin{abstract}
  The generic limit set of a cellular automaton is a topologically defined set of configurations that intends to capture the asymptotic behaviours while avoiding atypical ones. It was defined by Milnor then studied by Djenaoui and Guillon first, and by  T{\"o}rm{\"a} later. They gave properties of this set related to the dynamics of the cellular automaton, and the maximal complexity of its language. In this paper, we prove that every non trivial property of these generic limit sets of cellular automata is undecidable.
\end{abstract}

\section{Introduction}

Cellular automata (CA) are discrete dynamical systems defined by a local rule, introduced in the 40s by John von Neumann \cite{Neumann}. Given a finite alphabet $\A$, the global rule  on $\A^{\Z}$ is given by the synchronous application of the local one at every coordinate. They can be seen as models of computation, dynamical systems or many phenomena from different fields, providing links between all of these \cite{Hedlund-1969,Kari-survey}.

The asymptotic behaviour of CA has been studied a lot, mainly using the definition of limit set: the set of points that can be observed arbitrarily far in time. In particular concerning the complexity of this set: it can be non-recursive, the nilpotency problem is undecidable and there is Rice's theorem on properties of the limit set of CA \cite{cpy-1989,Kari-1992,Kari-1994}. Rice's theorem states that every nontrivial property of the limit set of CA is undecidable. Other definitions were introduced in order to restrain to typical asymptotic behaviour. Milnor proposed the definition of likely limit set and generic limit set in \cite{mil-1985} in the more general context of dynamical systems. While the likely limit set is defined in the measure-theoretical world, the generic limit set is a topological variant. Djenaoui and Guillon proved in \cite{djegui-2018} that both are equal for full-support $\sigma$-ergodic measures in the case of CA.

The generic limit set is the smallest closed subset of the fullshift $\SZ$ containing all limit points of all configurations taken in a comeager subset of $\SZ$. Djenaoui and Guillon studied the generic limit set in \cite{djegui-2018}, proving results on the structure of generic limit sets related to the directional dynamics of CA. They also provide a combinatorial characterization of the language of the generic limit sets and examples of CA with different limit, generic limit and $\mu$-limit sets. The latter was introduced in \cite{km-2000} by K\r{u}rka and Maass as another measure-theoretical version of limit set.

The $\mu$-limit set is determined by its language which is the set of words that do not disappear in time, relatively to the measure $\mu$. Amongst the results on the $\mu$-limit set, it was proved in \cite{bdpst} that the complexity of the language is at the level $3$ of the arithmetical hierarchy ($\Sigma_3^0$), with a complete example, it was also proved that the nilpotency problem is $\Pi_3^0$-complete. Rice's theorem also holds stating that each nontrivial property has at least $\Pi_3^0$ complexity. A slightly different approach led Hellouin and Sablik to similar results on the limit probability measure in \cite{Hellouin-Sablik-2013}.

In \cite{tor-2020}, T{\"o}rm{\"a} proved computational complexity results on the generic limit sets, in particular an example of a CA with a $\Sigma_3^0$-complete generic limit set, and constraints on the complexity when the dynamics of the CA is too simple on the generic limit set.

In this paper, we prove  Rice's theorem on generic limit sets combining ideas from \cite{Kari-1994} and \cite{bdpst}.

\section{Definitions}
In this paper, we consider the countable set $\mathcal{Q}=\{q_0,q_1,q_2,\dots\}$. Every finite alphabet will be a finite subset of $\mathcal{Q}$. Given a finite alphabet $\Sigma\subseteq \mathcal{Q}$ and a radius $r\in\N$, a local rule is a map $\delta:\Sigma^{2r+1}\to \Sigma$ and a \emph{cellular automaton} $\mathcal{F}:\Sigma^{\Z}\to \Sigma^{\Z}$ is the global function associated with some local rule $\delta$: for every $c\in \Sigma^{\Z}$ and every $i\in\Z$, $\mathcal{F}(c)_i=\delta(c_{i-r},c_{i-r+1},\dots,c_{i+r})$. We call \emph{configurations} the elements of $\Sigma^{\Z}$. The orbit of an initial configuration $c$ under $\mathcal{F}$ is called a \emph{space-time diagram}. Time goes upward in the illustrations of this paper.

Define the Cantor topology on $\Sigma^{\Z}$ using the distance $d(c,c')=\frac{1}{2^i}$ where $i=\min\{j\in\N, c_j\neq c'_j\textrm{ or }c_{-j}\neq c'_{-j}\}$. For any word $w\in \Sigma^*$, denote $|w|$ the length of $w$ and  $[w]_i=\{c\in \Sigma^{\Z}:\forall k< |w|,c_{i+k}=w_k\}$ the associated \emph{cylinder set}, which is a clopen set.

Denote $\sigma$ the shift on $\Sigma^{\Z}$, which is the CA such that $\forall c\in \Sigma^{\Z}, \forall i\in\Z, \sigma(c)_i=c_{i+1}$. A \emph{subshift} is a closed $\sigma$-invariant subset of $\Sigma^{\Z}$. A subshift can be equivalently defined by the set of forbidden words, in this case a subshift is the set of configurations that do not belong to any $[w]_i$ where $w$ is forbidden.

In this paper, a Turing machine works on a semi-infinite (to the right) tape, with a finite alphabet $\mathcal{A}$ containing a blank symbol $\bot$. It has one initial state $q_0$ and one final state $q_f$. At each step of the computation, the head of the machine reads the symbol at the position on the tape to which it points, and decides the new symbol that is written on the tape, the new state it enters, and its move (one cell at most). It can be simulated by a CA using states that can contain the head of the machine and the tape alphabet. We will here only simulate machines in a finite space in which there is only one head.

\subsection{Limit sets of cellular automata}

Different definitions of the asymptotic behavior of a CA have been given. The most classical one is the \emph{limit set} $\Omega_{\mathcal{F}}=\bigcap_{t\in\N}\mathcal{F}^t(\Sigma^{\Z})$ of a CA $\mathcal{F}$, that is the set of configurations that can be seen arbitrarily late in time. For any subset $X\subseteq \Sigma^{\Z}$,  define $\omega(X)$ as the set of limit points of orbits of configurations in $X$: $c\in\omega(X)\Leftrightarrow \exists c'\in X, \liminf_{t\to\infty}d(\mathcal{F}^t(c'),c)=0$. The set  $\omega(\Sigma^{\Z})$ is called the \emph{asymptotic set} of $\mathcal{F}$.

A subset $X\subseteq \Sigma^{\Z}$ is said to be \emph{comeager} if it contains a countable intersection of dense open sets. It implies in particular that $X$ is dense (Baire property). 

For $X\subseteq \Sigma^{\Z}$, define the \emph{realm of attraction} $\mathcal{D}(X)=\{c\in\Sigma^{\Z}:\omega(c)\subseteq X\}$. The \emph{generic limit set} $\tilde{\omega}(\mathcal{F})$ of $\mathcal{F}$ is then defined as the intersection of all closed subsets of $\Sigma^{\Z}$ whose realms of attraction are comeager.

The following two examples show differences between all these sets, they were already presented in \cite{djegui-2018}.
\begin{example}[The Min CA]
  Consider the CA $\mathcal{F}$ of radius $1$ on alphabet $\{0,1\}$ whose local rule is $(x,y,z)\mapsto \min{(x,y,z)}$. The state $0$ is spreading, that is, every cell that sees this state will enter it too. A space-time diagram of the MIN CA is represented in Figure~\ref{fig:min}.

  We have:
  \begin{itemize}
  \item $\Omega_{\mathcal{F}}=\{c\in\{0,1\}^{\Z}:\forall i\in\Z,k\in\N^*, c\notin [10^k1]_i]\}$;
  \item $\genF=\{0^{\Z}\}$ and it is equal to the $\mu$-limit set for a large set of measures containing every non degenerate Markov measure.
  \end{itemize}

  \begin{figure}
  \begin{center}
    \includegraphics[width=\textwidth]{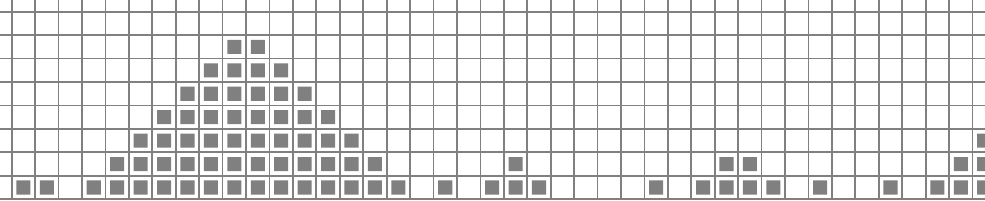}
    \caption{\label{fig:min}Some part of a space-time diagram of the Min CA, $0$ is represented by the white state and $1$ by the black state.}
    \end{center}
\end{figure}

\end{example}

\begin{example}[Gliders]
  Consider the CA $\mathcal{F}$ of radius $1$ on alphabet $\{0,>,<\}$. The states $<$ and $>$ are respectively speed $-1$ and $1$ signals over a background of $0$s. When a $<$ and a $>$ cross, they both disappear.  A space-time diagram of this CA is represented in Figure~\ref{fig:gliders}. For a complete description of the rule, see for example \cite[Example 3]{km-2000}.

  We have:
  \begin{itemize}
  \item $\Omega_{\mathcal{F}}=\{c\in\{0,<,>\}^{\Z}:\forall i\in\Z,k\in\N, c\notin [<0^k>]_i]\}$;
  \item $\genF=\Omega_{\mathcal{F}}$;
    \item the $\mu$-limit set depends here of $\mu$. With $\mu$ the uniform Bernoulli measure, it is $\{0^{\Z}\}$. If $\mu$ is Bernoulli with a bigger probability for $<$ than for $>$, then the $\mu$-limit set is  $\{\{0,<\}^{\Z}\}$.
  \end{itemize}

  \begin{figure}
  \begin{center}
    \includegraphics[width=.8\textwidth]{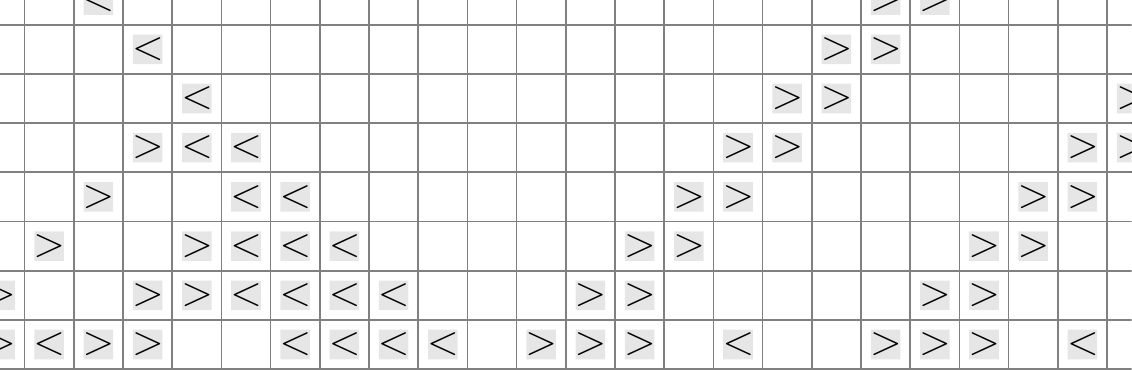}
    \caption{\label{fig:gliders} The $<$ and $>$ states of the Gliders CA are particles going in different directions and annihilating each other when they cross.}
    \end{center}
\end{figure}

\end{example}

\subsection{Preliminary properties of generic limit sets of CA}

Many properties of generic limit sets were proved either in \cite{mil-1985} or in \cite{djegui-2018} for the particular case of CA.
\begin{proposition}[Prop 4.2 of \cite{djegui-2018}]
  Given a CA $\mathcal{F}$, the realm of attraction of $\tilde{\omega}(\mathcal{F})$ is comeager.
\end{proposition}

\begin{proposition}[Prop 4.4 of \cite{djegui-2018}]
  Given a CA $\mathcal{F}$, $\tilde{\omega}(\mathcal{F})$ is a subshift.
\end{proposition}
Note that the limit set of a CA is also a subshift whereas the asymptotic limit set may not be.

\begin{proposition}[Cor 4.7 of \cite{djegui-2018}]
  Given a CA $\mathcal{F}$ on alphabet $\Sigma$, $\tilde{\omega}(\mathcal{F})=\Sigma^{\Z}\Leftrightarrow \mathcal{F}\textrm{ is surjective}$.
\end{proposition}

The last result of this section comes from Remark 4.3 of \cite{djegui-2018} and is reformulated as Lemma 2 of \cite{tor-2020}:
\begin{lemma}\label{lem:combchar}
  Let $\mathcal{F}$ be a CA on $\Sigma^{\Z}$. A word $s\in \Sigma^*$ occurs in $\tilde{\omega}(\mathcal{F})$ if and only if there exists a word $v\in \Sigma^*$ and $i\in\Z$ such that for all $u,w\in \Sigma^*$, there exist infinitely many $t\in\N$ with $\mathcal{F}^t([uvw]_{i-|u|})\cap [s]\neq\emptyset$.
\end{lemma}

The word $v$ is said to \emph{enable} $s$.

\section{General structure of the construction}
\label{sec:constr}
The proof of the main result of this paper relies on a construction already presented in~\cite{Delacourt-2011,bdpst,Hellouin-Sablik-2013}. The present section contains the description of this tool. The idea is to erase most of the content of the initial configuration and start a protected (hence controled) and synchronized evolution. Of course, to ensure that  this property holds for any configuration, one needs strong constraints on the dynamics of the CA. Here, we also want to allow a wide variety of dynamics, hence this property shall hold for almost every initial configuration. In the above-cited articles, it was true for $\mu$-almost every configuration, and here we will use a topological variant.

A brief description of this CA $\mathcal{F}$ follows. Its radius should be at least $2$.

\subsection{Overview}
Some particular state $\sti{}\in\Si$ can only appear in the initial configuration: there is no rule that produces it. The states \sti{} will trigger the desired evolution. In order to avoid having to deal with anything unwanted on the initial configuration (like words produced by the evolution of the CA placed in a wrong context), we add  a mechanism that cleans the configuration from anything that is not produced by \sti{}. This is achieved through the propagation of large signals that have the information of the time passed since a \sti{} state produced it, that is their age. Then, when two such signals going in opposite directions meet, they compare their ages and only the younger survives.

With this trick, any configuration that contains infinitely many \sti{} on both sides will ultimately be covered by protected areas. The \sti{} states also transform into  $\sts{}\in\Si$ states, and we consider the words in the space-time diagram that are delimited by \sts{} states produced by \sti{} states, we call them segments.  The dynamics of the CA inside a segment only depends on its size. In particular, the simulation of the computation of a given Turing machine can be started on each \sts{} state when it appears.

A close construction with a more precise and complete description can be found in \cite[Section 3.1]{bdpst}.

\subsection{Initialization and counters}
The state \sti{} can only appear in the initial configuration: it is not produced by any rule and it disappears immediately. Consider a cell at coordinate $i$ that contains a $\sti{}$ state in the initial configuration. On each side of the \sti{} state, two signals are sent at speed $s_f$ and $s_b$ to the right and symmetrically to the left. The fastest one (speed $s_f$) erases everything it encounters except for its symmetrical counterpart. Each couple of signals is seen as one counter whose value is encoded by the distance $\lfloor k(s_f-s_b)\rfloor$ after $k$ steps of the CA. The key point is that, at any time, the value of a counter is minimal exactly for counters generated by a \sti{} state.

When two counters meet, they compare their values without being affected until the comparison is done. The comparison process is done via signals bouncing on the borders of the counters. The speed of these inner signals is greater than the speeds ($s_f$ and $s_b$) of the border signals. As the value is encoded by the distance between border signals, it is a geometric comparison illustrated in Figure \ref{fig:comp}. If one counter is younger than the other one, the older one is deleted (the right one in Figure \ref{fig:comp}). If they are equal, both counters, that is the $4$ signals, are deleted.

\begin{claim}\label{cla:crosscount}
  For any configuration $c$ where $\sti{}$ occurs, and any coordinate $i\in\Z$, denote $d_i=\min\{|i-j|:c_j=\sti\}$. Then for any $t>s_bd_i$ (where $s_b$ is the speed of the inner border of the counter), $\mathcal{F}^t(c)_i$ does not contain a counter state. 
\end{claim}
\bclaimproof{}
  Each sequence of consecutive \sti{} states creates a left counter at its left extremity and a right counter at its right extremity. They all share a common age which is the minimal one, hence they cannot be crossed by another counter. Thus, at most one of the youngest counters can cross cell $i$. And due to the speed of the inner border of the counters, this is done after  $s_bd_i$ steps.
\eprf

\begin{figure}
  \begin{center}
    \includegraphics[width=.35\textwidth]{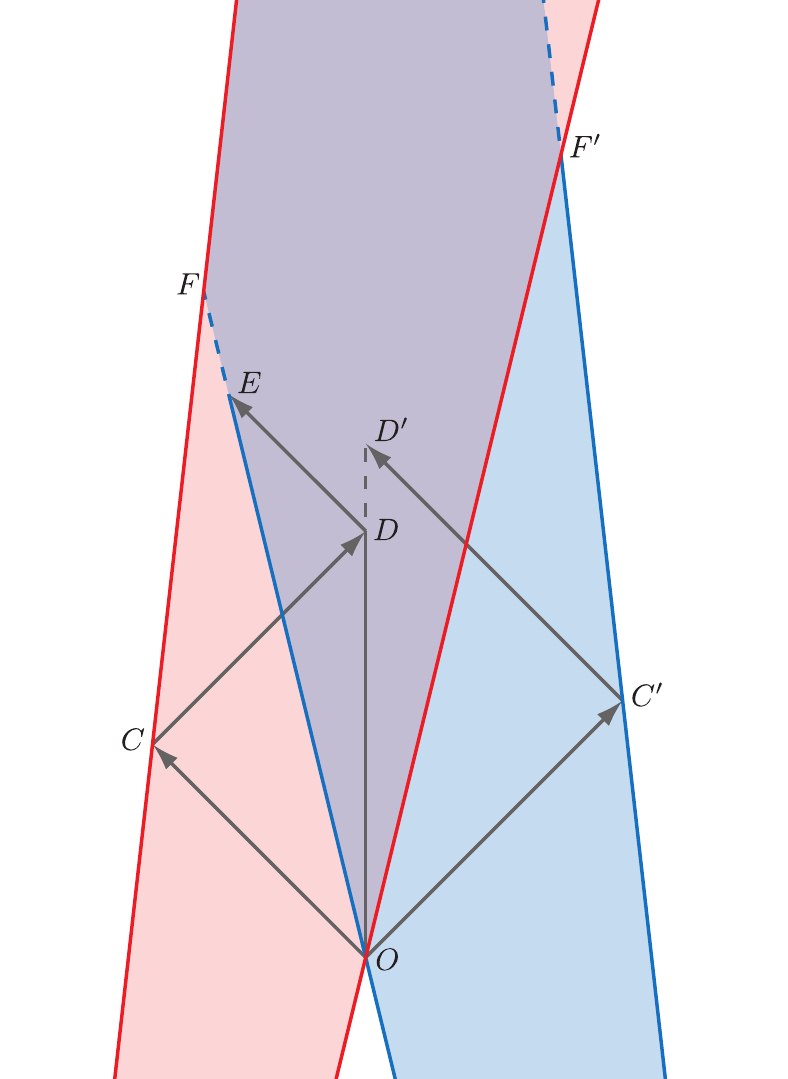}
    \caption{\label{fig:comp}When counters meet in $O$, signals move at speed $1$ towards the borders of the counters that they reach at points $C$ and $C'$. They bounce back until they cross the sign left at point $O$. The one that arrives first has crossed the most narrow (hence youngest) one. It bounces once again to erase the opposite counter whose border is reached at point $E$.}
    \end{center}
\end{figure}

Last rule of this construction: every $\sti{}$ state that is not surrounded by other $\sti{}$ states on both sides is replaced by a $\sts{}$ state after it gave birth to the counters. Figure \ref{fig:count} shows how a typical initial configuration evolves.

For any time $t\in\N$ and any configuration $c$, we call \emph{segment} a set of consecutive cells from coordinate $i$ to $j$ in $\mathcal{F}^t(c)$ with $i,j\in\Z$ such that:\begin{itemize}
\item $\mathcal{F}^t(c)_i=\sts=\mathcal{F}^t(c)_j$
\item for every $i<k<j$,  $ \mathcal{F}^t(c)_k\neq \sts$
\item $c_i=\sti$ and $c_j=\sti$.
\end{itemize}

\begin{figure}
  \begin{center}
    \includegraphics[width=.8\textwidth]{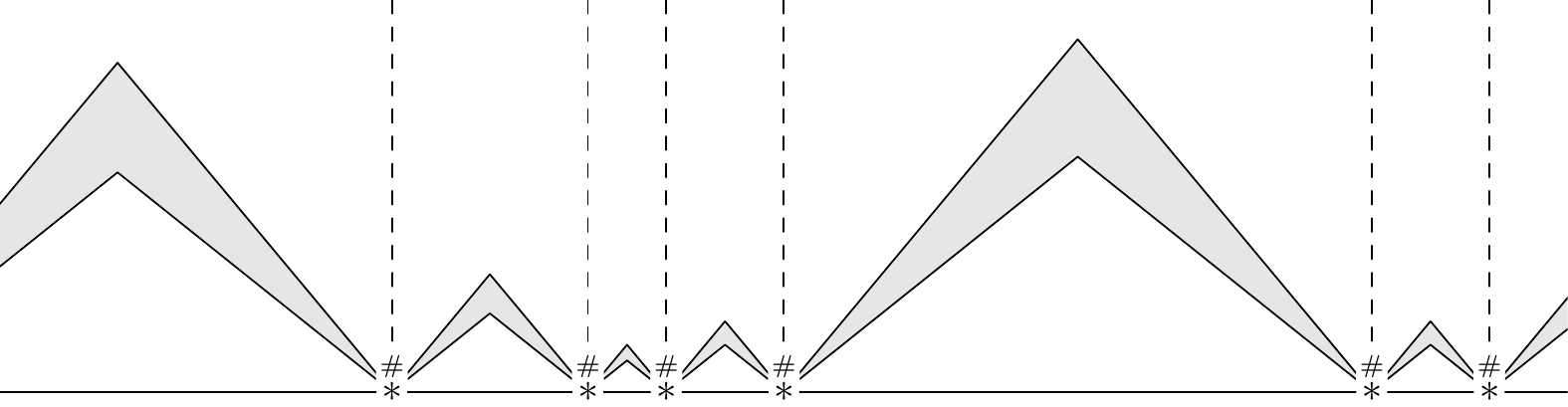}
    \caption{\label{fig:count}Starting from a configuration containing infinitely many \sti{} states on the left and on the right, the $\sti{}$ states generate counters (filled in grey) on both sides that erase everything but another counter going in the opposite direction. These counters eventually meet their opposite and disappear after comparing their ages, hence remain an immaculate configuration with $\sts{}$ states in some positions.}
    \end{center}
\end{figure}

Note that if the radius of the CA can be arbitrarily large, any choice of speeds $s_f>s_b$ can be made.

\begin{claim}\label{cla:countspeed}
  For any $s\in\mathbb{Q}$, there exists a CA implementing such a construction with speed $s_b>s$ (and hence $s_f$).
\end{claim}
\bclaimproof{}
  A big enough radius allows fast enough signals to perform the comparison of counters in due time.
\eprf

\section{Rice's theorem}
Following the steps of the historical proof of Rice and concerning CA, the theorems on limit sets in \cite{Kari-1994} and $\mu$-limit sets in \cite{Delacourt-2011}, we first define properties of generic limit sets of CA, then prove that every non trivial such property is undecidable.

The CA used in \cite{Delacourt-2011} to prove Rice's theorem for $\mu$-limit sets also has the general structure presented in the previous section. The difference lies in what is done inside segments. In the case of $\mu$-limit sets (regardless of the choice of $\mu$), it is possible to dedicate a small \emph{technical} space inside segments to any activity that shouldn't appear in the $\mu$-limit set, as long as this space tends to disappear in density. This is achieved through larger and larger segments. Nothing prevents the states of this technical space to appear in the generic limit set.

\subsection{Properties of generic limit sets of CA}

A property of the generic limit set of CA is a set of subshifts and we say that a generic limit set have this property if it belongs to this set. This way, it depends only on the generic limit set: if two CA have the same generic limit set, this common generic limit set either has or not the property. As mentionned earlier, we consider the countable set $\mathcal{Q}=\{q_0,q_1,\dots\}$, and every alphabet is a finite subset of $\mathcal{Q}=\{q_0,q_1,\dots\}$.

\begin{definition}
A property $\mathcal{P}$ of generic limit sets of cellular automata is a subset of the powerset $\mathscr{P}(\mathcal{Q}^{\Z})$. A generic limit set of some cellular automaton is said to have property $\mathcal{P}$ if it is in $\mathcal{P}$.
\end{definition}

Note that many sets that are not subshifts can belong to a property $\mathcal{P}$, as every generic limit set is a subshift, they do not matter. In particular, every property that does not contain a subshift is equivalent to the empty property that no generic limit set has. A property is said to be \emph{trivial} when either it contains all generic limit sets or none. The most natural example of a non trivial property is the \emph{generic nilpotency}, which is given by the family $\{\{q_i^{\Z}\},i\in\N\}$.

This definition prevents confusions between properties of generic limit sets and properties concerning generic limit sets. For example the property containing every fullshift on finite alphabets is not surjectivity, since the generic limit set of a CA on alphabet $\Sigma$ could be a fullshift on a strictly smaller alphabet. Hence surjectivity is not a property of generic limit sets even if being surjective is equivalent to having a full generic limit set.
.

\subsection{The theorem}

\begin{theorem}
Every non trivial property of the generic limit sets of CA is undecidable.
\end{theorem}

This section is dedicated to the proof of Rice's theorem. It is a many-one (actually one-one) reduction from the Halting problem on empty input for Turing machines.
  Take a non trivial property $\mathcal{P}$ of generic limit sets of CA. Assume for example that $\mathcal{P}\cap\{\{q_k^{\Z}\},k\in\N\}$ is infinite (the other case leads to a symmetric proof). As $\mathcal{P}$ is non trivial, it is possible to choose $q_n\in\mathcal{Q}$ and a CA $\mathcal{F}_1$ such that $\tilde{\omega}(\mathcal{F}_1)\notin \mathcal{P}$ and $q_n\notin \Sigma_1$ where $\Sigma_1$ is the alphabet of $\mathcal{F}_1$. Denote now $\mathcal{F}_0$ the CA on alphabet $\{q_n\}$ whose local rule always produces $\{q_n\}$. Hence $\tilde{\omega}(\mathcal{F}_0)=\{q_n^{\Z}\}\in\mathcal{P}$.

  For any Turing machine $M$, we produce a CA $\mathcal{F}_M$ such that:
  \begin{itemize}
  \item if $M$ eventually halts on empty input, the generic limit set of $\mathcal{F}_M$ is $\{q_n^{\Z}\}$;
  \item if $M$ never halts on empty input, then the generic limit set of $\mathcal{F}_M$ is $\tilde{\omega}(\mathcal{F}_1)$.
  \end{itemize}

  \subsubsection{Construction of $\mathcal{F}_M$}

  The CA $\mathcal{F}_M$ contains two layers, one for each of the main tasks.  Denote $\pi_1$ and $\pi_2$ the projections on the first and second layer. The first layer uses alphabet $\Sigma_0$ and it implements the construction described in Section~\ref{sec:constr}. Denote $\_$ the blank state of $\Sigma_0$. The second layer simulates the CA $\mathcal{F}_1$. In some cases, the first layer can be erased, we also add a state $q_n$, hence the alphabet of $\mathcal{F}_M$ is $\Sigma=(\Sigma_0\times \Sigma_1)\cup\{q_n\}\cup \Sigma_1$.

  The set $\Sigma_0\times \Sigma_1$ can be mapped to a subset of $\mathcal{Q}\setminus (\{q_n\}\cup \Sigma_1)$ to ensure that $\Sigma\subset \mathcal{Q}$. For the clarity of the presentation, we will denote the elements of $\Sigma_0\times \Sigma_1$ as couples.

  The idea is to let $\mathcal{F}_1$ compute on the second layer (or by itself if the first layer has been erased), while computation on the first layer will either lead to erase this layer or generate a $q_n$ state that will be spreading (erasing everything but counters) over the whole configuration.

  On the first layer, once a $\sts{}$ state appears (from a $\sti{}$ state), a simulation of $M$ is started on its right. In the general case, another $\sts{}$ state exists further on the right, in which case this simulation takes place in a segment. We will show later that the other case is irrelevant when considering the generic limit set. The simulation evolves freely except if it is blocked by the inner border of a counter, if this happens the simulated Turing head waits until it has enough space to make one more step.
  A binary counter is started in parallel to the right of the $\sts{}$ state.

  The simulation inside a segment should always be finite, it can be interrupted for one of the following reasons.
  \begin{itemize}
    \item The simulation of $M$ halts (because $M$ reaches a final state). Then the state $q_n$ is written, erasing both layers of $\mathcal{F}_M$. This state spreads to both of its neighbors erasing everything, even the $\sts{}$ states, except for the inner and outer borders of the counters of the construction of Section~\ref{sec:constr}.
    \item It reaches  a $\sts{}$ on its right. That is there is not enough space inside the segment and the simulation is aborted. The first layer content of the segment will be erased as explained later.
    \item The counter reaches another $\sts{}$ state. The time allowed for the simulation is over and the simulation is aborted. This third case is necessary to avoid problems due to a loop of the Turing machine in a finite space.
  \end{itemize}

  The states used for the simulation should not appear in the generic limit set, hence they have to be erased once the simulation halts or is aborted. In the first case, the state $q_n$ is written in every cell. In the second case, the first layer only is erased. For the same reason, the \sts{} state has to be erased when the simulation is over in both the segments it delimits.
  
  If the simulation is aborted (due to lack of space or end of the allowed time in the segment), an \emph{abortion signal} is sent in both directions that erases everything of the first layer (except outer or inner border of counters) until it reaches a $\sts{}$ state. A $\sts{}$ state that receives such an abortion signal transforms into a $\stsp{}$ state. If a $\stsp{}$ state receives an abortion signal, it disappears. The point is to ensure that the abortion signals do not travel too far: if the first abortion signal deletes the $\sts$ state on the side of the segment, then the one arriving from the other side will cross. This could lead to the presence of abortion signals in the generic limit set.

  Figure~\ref{fig:genpic} is a schematic view of the evolution of CA $\mathcal{F}_M$ on an ordinary initial configuration.
  
  \begin{figure}
  \begin{center}
    \includegraphics[width=.6\textwidth]{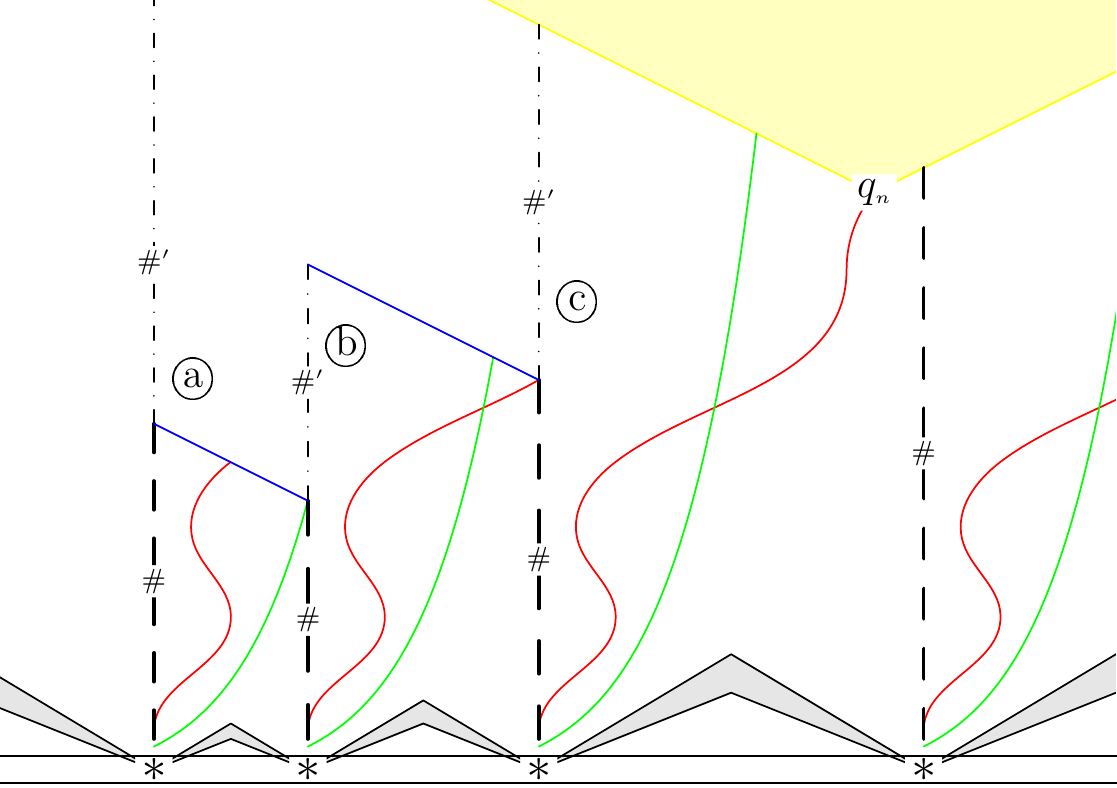}
    \caption{\label{fig:genpic}Starting from the cells in state \sti{} in the initial configuration, the counters (grey areas) protect everything above them. Segments are delimited by \sts{} states and in each of them a simulation of the computation of a Turing machine takes place (the red curve gives the position of the head). The green curve represents the extension of the binary counter used to limit the time of the simulation. In segment \textcircled{a}, the counter reaches the limit and an abortion signal is sent (blue). In segment \textcircled{b}, the head reaches the right of counter and the simulation is stopped with an abortion signal sent to the left. In segment \textcircled{c}, the Turing machine halts and the spreading state $q_n$ is written.}
    \end{center}
\end{figure} 

  \begin{claim}\label{cla:abort}
    There exists an increasing function $f:\N\to\N$ such that the computation of $M$ simulated in a segment of length $n$ either halts or is aborted before time $f(n)$.
  \end{claim}
  \bclaimproof{}
In a segment of length $n$, due to the binary counter, if the simulation of $M$ has not reached a final state after $2^n$ steps, the computation is aborted.
  \eprf

\subsubsection{Ensuring a sound computation on the second layer}
  
  The proof relies on the fact that, with most initial configurations:
  \begin{itemize}
  \item if $M$ halts, there will exist a large enough segment in which the computation has enough space and time to reach its end, thus producing state $q_n$ that erases everything.
  \item if $M$ does not halt, the computation will be eventually aborted in every segment and only the second layer will remain with a computation of  $\mathcal{F}_1$.
  \end{itemize}

  In order to ensure the second point, we need to deal with the case of $q_n$ states existing before the counters of Section~\ref{sec:constr} clean the configuration on the first layer. It can for example happen due to $q_n$ states on the initial configuration. In this case, the content of the second layer is lost. As it is impossible to control what happens outside the area protected by counters, the counters will not only stop the spreading of $q_n$ but also write a possible configuration for $\mathcal{F}_1$, thus deleting all data that does not descend from the cells containing $\sti$ in the first layer of the initial configuration.

  Let us assume for simplicity that the radius of $\mathcal{F}_1$ is $1$. For the rest of the proof of the theorem, denote $x_0$ some state of $\Sigma_1$. The space-time diagram of $\mathcal{F}_1$ with initial configuration $x_0^{\Z}$ is ultimately periodic, contains only uniform configurations and is entirely described by a finite sequence of distinct states $(x_0,x_1,\dots,x_p,\dots,x_{p+T},x_p)$. The counters will write the second layer of the configuration as if every information coming from outside the protected area (between counters) was obtained from the uniform initial configuration $x_0^{\Z}$:
  \begin{itemize}
  \item $x_t$ at step $t\leq p$;
  \item $x_{p+(t-p)\mod T}$ at step $t\geq p$.
  \end{itemize}

  As a finite amount of information is needed, the local rule of the CA $\mathcal{F}_M$ can be designed to do so. This is illustrated in Figure~\ref{fig:rewritesl}. As said in Claim \ref{cla:countspeed}, it is possible to use that construction with outer borders of counters moving at speed $1$.

  If the first layer contains $\sti{}$, the state on the second layer is not rewritten and is used for the simulation of $\mathcal{F}_1$.

\begin{figure}
  \begin{center}
    \includegraphics[width=.5\textwidth]{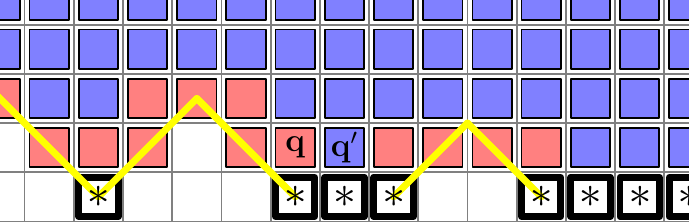}
    \caption{\label{fig:rewritesl}Partial representation of a space-time diagram of $\mathcal{F}_M$. The red cells are where the counters rewrite the second layer assuming that what does not come from a $\sti{}$ state is $x_0$. The blue cells are where the computation of $\mathcal{F}_1$ happens normally on the second layer. The yellow lines are the outer borders of counters, we assume here they have speed $1$ for the illustration. Denote $\delta_1$ the local rule of $\mathcal{F}_1$. Then $q'=\delta_1(x,y,z)$ which are its state ($y$) and the ones of its neighbors ($x$ and $z$) at time $0$. And $q=\delta_1(x_0,x,y)$.}
    \end{center}
\end{figure} 

To any initial configuration $x\in\Si^{\Z}$, corresponds a configuration in  $\Si_1^{\Z}$ where all the deleted data is replaced by $x_0$. Denote $\phi:\Si\to\Si_1$ such that:
  \begin{itemize}
  \item $\phi(\sti,x)=x$;
  \item $\phi(s,x)=x_0$ when $s\neq\sti$;
  \item $\phi(x)=x_0$ when $x\in \Si_1\cup\{q_n\}$.
  \end{itemize}
  It can be extended to words in $\Si^*$ and configurations in $\SZ$. 

\begin{claim}\label{cla:proj}
  Let $c$ be a configuration in $\Sigma^{Z}$ and $i\in\Z$ a coordinate such that there exists $j<i<k$ with $c_j=\sti=c_k$.
  Then for any $t>s_bd_i$ (as in Claim~\ref{cla:crosscount}),
  \[\pi_2\left(\mathcal{F}_M^t(c)_i\right)\in\left\{\mathcal{F}_1^t(\phi(c))_i,q_n\right\}\]
  We extend here $\pi_2$ as the identity to $\Si_1\cup\{q_n\}$.
\end{claim}
\bclaimproof{}
As $t>s_bd_i$, the cell at coordinate $i$ is in the protected area (above \sti{} states or counters) at time $t$. Then the second layer has been computed with the rule of $\mathcal{F}_1$ and the second layer of the configuration rewritten by counters into images of $\phi(c)$. The only way to interrupt the computation of $\mathcal{F}_1$ is to erase the cell and write $q_n$, hence the claim.
\eprf

\subsubsection{Proof of the theorem}

It remains to prove the next $2$ lemmas.
  \begin{lemma}
    If $M$ eventually halts on the empy input, then $\genFM=\tilde{\omega}(\mathcal{F}_0)\in\mathcal{P}$.
  \end{lemma}
  \begin{proof}
  Suppose that $M$ eventually halts on the empty input. Then there exists a large enough size $S$ such that the computation in any segment larger than $S$ has enough time and space to reach its end. Then the state $q_n$ appears and spreads at speed $1$ in both  directions except if it encounters an inner or outer border of a counter.

  If some state $s\in\Si$ occurs in $\genFM$ then according to Lemma~\ref{lem:combchar}, there exists a word $v$ that enables it when placed at position $i\in\Z$. Take now  $u=(\_\sti\_^S\sti\_,x_0^{S+4})$, $w$ the empty word and some $c\in[uvw]_{i-|u|}$. Counters are generated by the two $\sti{}$ states at coordinates $i-(S+3)$ and $i-2$, hence there exists $t_0\in\N$ such that at time $t_0$, the cell $0$ has  been crossed by counters generated by $\sti{}$ states. According to Claim~\ref{cla:crosscount}, it will not contain any state of outer or inner border of a counter anymore.
  Moreover, a segment is created between coordinates $i-(S+3)$ and $i-2$. As it is large enough, the state $q_n$ will be written at time $t_1\in\N$. Then it will spread and reach cell $0$ before time $t_1+\max(|i-(S+2)|,|i-2|)$ or $t_0$ if the inner border of a counter slows it down. This is illustrated by Figure~\ref{fig:proof1}. Hence there exists $t_2\in\N$ such that $\forall t\geq t_2, \mathcal{F}_M^t(c)\in [s]\Leftrightarrow s=q_n$. Thus $\genFM\subseteq  \{q_n\}^{\Z}$. As $\genFM$ cannot be empty, we have $\genFM=\{q_n\}^{\Z}=\tilde{\omega}(\mathcal{F}_0)$ and $\genFM\in\mathcal{P}$.

  \begin{figure}
  \begin{center}
    \includegraphics[width=.6\textwidth]{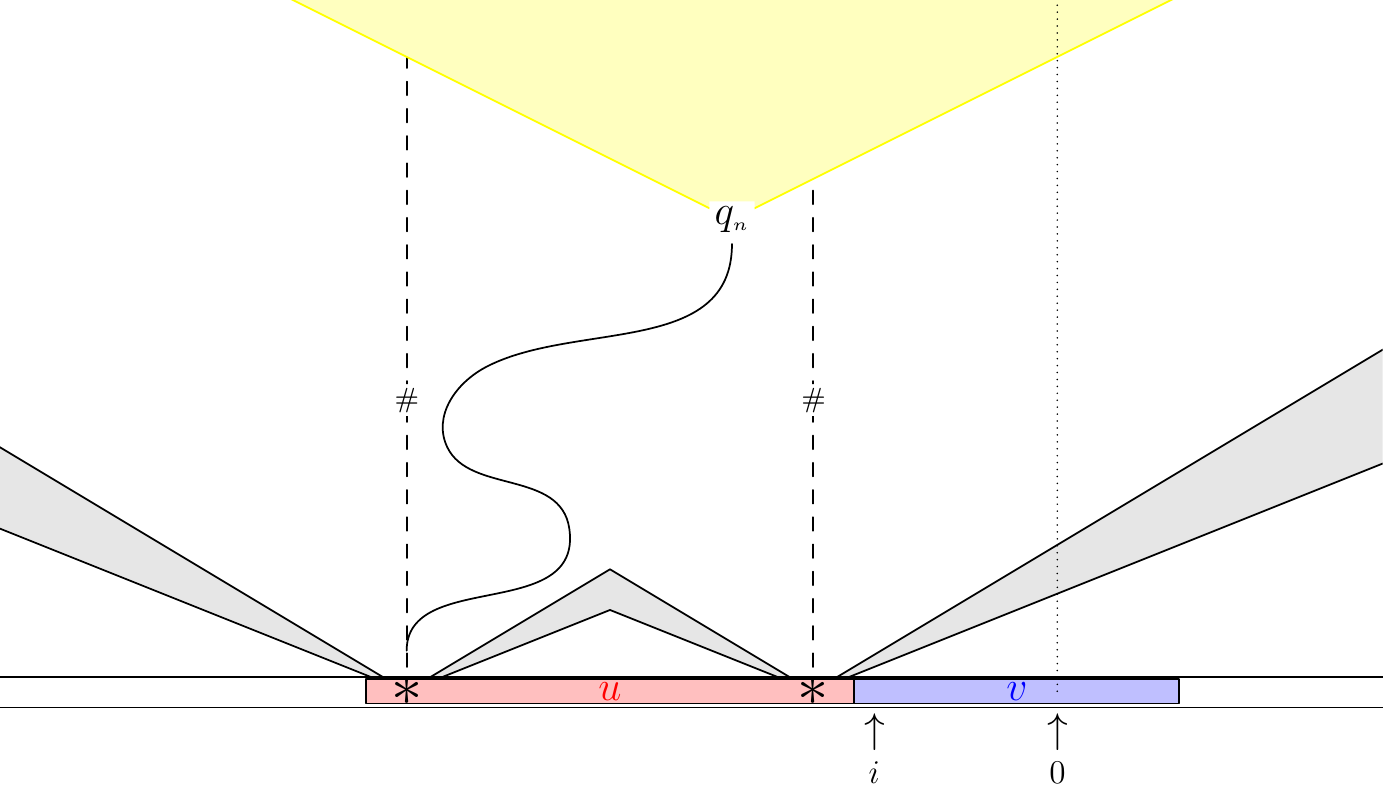}
    \caption{\label{fig:proof1}The word $v$ (in blue) is supposed to enable state $s$. Then for a good choice of $u$ (in red), a segment will simulate a computation of $M$ that eventually halts and produces $q_n$. This state spreads (in yellow) and eventually reaches coordinate $0$.  }
    \end{center}
\end{figure} 
  \end{proof}
  
  \begin{lemma}
If  $M$ never halts on the empty input, then $\tilde{\omega}(\mathcal{F}_M)=\tilde{\omega}(\mathcal{F}_1)\notin\mathcal{P}$.
  \end{lemma}
  \begin{proof}
    Suppose now that $M$ never halts on the empty input. We will show that $\tilde{\omega}(\mathcal{F}_1)=\tilde{\omega}(\mathcal{F}_M)$.

    First, we show that:
    \begin{claim}
      $\tilde{\omega}(\mathcal{F}_M)\subseteq\tilde{\omega}(\mathcal{F}_1)$
      \end{claim}
\bclaimproof{}

      Let $s$ be a word that occurs in $\tilde{\omega}(\mathcal{F}_M)$. According to Lemma~\ref{lem:combchar}, there exists a word $v$ that enables $s$ when placed at coordinate $i$.  As any word containing $v$ as a factor also enables $s$, we can choose $v$ such that $i<0$ and $i+|v|>|s|$.

  We prove that $v'=\phi(v)$ at coordinate $i$ enables $s$ for $\mathcal{F}_1$. To do so, we will use Lemma~\ref{lem:combchar}. Take $u',w'\in \Si_1^{*}$ and denote $u=(\sti\_^{|u'|-1},u')$ and $w=(\_^{|w'|-1}\sti,w')$. Denote $n=|uvw|$, $T\geq \max(s_bn,f(n)+n)$ (where $s_b$ is  the speed of inner borders of counters), $z_1=i-|u|$ and $z_2=i+|vw|$. Apply Lemma~\ref{lem:combchar} with $\mathcal{F}_M$, $v$, $u$ and $w$. For infinitely many times $t$, there exist a configuration $c\in [uvw]_{i-|u|}$ such that $\mathcal{F}_M^t(c)\in [s]$. Using Claim~\ref{cla:proj} with cells at coordinates $z_1$ and $z_2$ containing state $\sti$, we get that for any  $t>T$,
\[\forall z_1\leq j\leq z_2,\pi_2\left(\mathcal{F}_M^t(c)_j\right)\in\left\{\mathcal{F}_1^t(\phi(c))_j,q_n\right\}\]
That is : $\pi_2(s)=\pi_2\left(\mathcal{F}_M^t(c)_{[0,|u|-1]}\right)\in\left\{\mathcal{F}_1^t(\phi(c))_{[0,|u|-1]},q_n^{|u|}\right\}$.

Due to the $\sti$ states placed at coordinates $z_1$ and $z_2$, we can also apply Claim~\ref{cla:abort} and we get that the computation is finished in any segment between coordinates $z_1$ and $z_2$ at time $f(n)$. After $n$ more steps, the potential abortion signals have reached the borders and every cell between coordinates $z_1$ and $z_2$ contains a state in $\Si_1\cup\{q_n\}$. Moreover, as these cells belonged to a segment in the protected area, and since $M$ never halts on the empty input, this state cannot be $q_n$. Hence $s\in\Si^*$ and as $\pi_2$ is the identity on $\Si$, necessarily $s=\pi_2(s)=\mathcal{F}_1^t(\phi(c))_{[0,|u|-1]}$.

As $\phi(c)\in[u'v'w']_{i-|u'|}$ and as $\mathcal{F}_1^t(\phi(c))_{[0,|u|-1]}=s$ for infinitely many times $t$, Lemma~\ref{lem:combchar} allows to conclude that $v'$ enables $s$ that is $v'$ occurs in $\tilde{\omega}(\mathcal{F}_1)$.

\eprf

Then we prove the opposite:
\begin{claim}
      $\tilde{\omega}(\mathcal{F}_1)\subseteq\tilde{\omega}(\mathcal{F}_M)$
      \end{claim}
\bclaimproof{}
  Let $s$ be a word that occurs in $\tilde{\omega}(\mathcal{F}_1)$. According to Lemma~\ref{lem:combchar}, there exists a word $v$ that enables it when placed at coordinate $i$. We prove that $v'=(\_^{|i|}\sti^{|v|}\_^{|i|+|s|},x_0^{|i|}vx_0^{|i|+|s|})$ at coordinate $i-|i|$ enables $s$ for $\mathcal{F}_M$.

  For any $u',w'\in\Si^*$, denote $n=|u'v'w'|$. Let $T\geq \max(s_bn,f(n)+n)$ (where $s_b$ is still the speed of inner borders of counters) and denote
  \begin{itemize}
  \item $u=\phi(\pi_2(u'))x_0^{|i|}$;
  \item $w=x_0^{|i|+|s|}\phi(\pi_2(w'))$.
  \end{itemize}
  As $v$ enables $s$ for $\mathcal{F}_1$, there exists $c\in[uvw]_{i-|u|}$ and $t\geq T$ such that $\mathcal{F}_1^t(c)\in [s]$. We can write $c$ as $c_-uvwc_+$ where $c_-$ and $c_+$ are semi-infinite configuration in $^{\omega}\Si_1$ and $\Si_1^{\omega}$ respectively. Define $c'=(^{\omega}\sti{},c_-)u'v'w'(\sti{}^{\omega},c_+)\in [u'v'w']_{i-|i|-|u'|}$, we will prove that $\mathcal{F}_M^t(c')\in [s]$.
 First, note that $c=\phi(\pi_2(c'))$. Then using Claim~\ref{cla:proj}, we have that for every $j\in[|i-|i|,i+|i|+|s|]$:
  \[\pi_2\left(\mathcal{F}_M^t(c')_j\right)\in\left\{\mathcal{F}_1^t(c)_j,q_n\right\}\]
  As $t\geq T\geq s_bn$ and $M$ does not halt on the empty input, $\pi_2\left(\mathcal{F}_M^t(c')\right)_j\neq q_n$. And as $t\geq T\geq f(n)+n$, the computation is aborted in every segment fully located between coordinates $|i-|i|$ and $i+|i|+|s|$ before step $f(n)$. After $n$ more steps, the first layer of these segments is erased, in particular for coordinates $j$ with $0\leq j < |s|$. Hence $\mathcal{F}_M^t(c')_{[0,|s|-1]}=\pi_2\left(\mathcal{F}_M^t(c')\right)_{[0,|s|-1]}=\mathcal{F}_1^t(c)_{[0,|s|-1]}=s$ and $s\in\genFM$.
  \eprf
  \end{proof}

  The last two lemmas show that $M\mapsto \mathcal{F}_M$ is a reduction from the Halting problem of Turing machines on empty input to the problem of decision of $\mathcal{P}$.

\section{Conclusion and perspectives}
We proved Rice's theorem for generic limit sets of CA, which means that for example generic nilpotency is undecidable. In the case of limit sets and $\mu$-limit sets, the nilpotency problem has the lowest complexity in the arithmetical hierarchy amongst properties of limit or $\mu$-limit sets ($\Sigma_1^0$-complete for limit sets and $\Pi_3^0$-complete for $\mu$-limit sets). It may be the case once more  for generic limit sets. Lemma~\ref{lem:combchar} gives a $\Pi_3^0$ upper bound on the complexity of generic nilpotency and T{\"o}rm{\"a} suggests in \cite{tor-2020} that the exact complexity could be obtained using a construction close to the one presented in \cite{bdpst} or in the present paper. One might think that another version of Rice's theorem could be deduced where the lower bound of complexity on non trivial properties of generic limit sets is higher than $\Sigma_1^0$. 

Using again constructions of \cite{bdpst}, one can certainly prove properties similar to the ones obtained on $\mu$-limit sets in the same paper, but also build examples to show that the languages of $\mu$-limit set and generic limit set can have totally distinct complexities like $\Sigma_3^0$-complete versus a full-shift.

\bibliography{biblio}

\end{document}